\documentclass{article}       
\usepackage{soul}
\usepackage[table,xcdraw]{xcolor}
\usepackage{makeidx}  
\usepackage[english]{babel}
\usepackage[utf8]{inputenc}
\usepackage{amsmath}
\usepackage{graphicx}
\usepackage[colorinlistoftodos]{todonotes}
\usepackage{amssymb}
\usepackage{color}
\usepackage{graphicx}
\usepackage{amsmath}
\usepackage{amssymb}
\usepackage{subfigure}
\usepackage{latexsym}
\usepackage{color}
\definecolor{blue}{rgb}{0,0,1}

\usepackage[algo2e,ruled,vlined]{algorithm2e}
\newcommand{\cech}{\mbox{\it \v{C}}}
\newcommand{\VR}{\mbox{\it VR}}
\newcommand{\proof}{\noindent{\bf Proof. }}
\newcommand{\qed}{\hfill $\Box$}

\newtheorem{alg}[section]{Procedure}
\newtheorem{remark}[section]{Remark}
\newtheorem{theorem}[section]{Theorem}
\newtheorem{definition}[section]{Definition}
\newtheorem{proposition}[section]{Proposition}

\begin{document}

\title{Persistent Entropy for Separating Topological Features from Noise in Vietoris-Rips Complexes 
\thanks{Authors are partially supported by  Spanish Government
under grant MTM2015-67072-P (MINECO/FEDER, UE).}
}


\author{Nieves Atienza$^1$,  Rocio Gonzalez-Diaz$^1$, 
       Matteo Rucco$^2$ 
\\\\
$^1$ 
              Depto. Mat. Aplic. I, E.T.S.I. Informatica, \\
University of Seville, Spain \\
              email: $\{$natienza,rogodi$\}$@us.es          
         \\
           $^2$
             School of Science and Tech., Computer Science Division,\\
University of Camerino, Italy\\
email: matteo.rucco@unicam.it 
}

\maketitle

\begin{abstract}
Persistent homology studies the evolution of $k$-dimensional holes along a nested sequence of simplicial complexes (called a filtration). The set of bars (i.e. intervals)  representing birth and death times of  $k$-dimensional holes along such sequence is called the persistence barcode. $k$-Dimensional holes with short lifetimes are informally considered to be ``topological noise'', and those with  long lifetimes are considered to be ``topological features'' associated to the   filtration.
  {\it Persistent entropy} is defined  as the Shannon entropy of the persistence barcode of a given filtration. In this paper we present  new important properties of  persistent entropy of $\cech$ech and Vietoris-Rips filtrations.  Among the properties, we put a focus on the stability theorem that allows to use persistent entropy for comparing persistence barcodes. Later, 
 we derive a simple method for separating topological noise from features in Vietoris-Rips filtrations.

\end{abstract}
\section{Introduction}

Topology is the branch of mathematics that studies shapes and maps among them. From the algebraic definition of topology a new set of algorithms have been derived. These algorithms are identified with ``computational topology'' or often pointed out as Topological Data Analysis (TDA) and are used for investigating high-dimensional data in a quantitative manner.

\begin{figure}[t!]
	\centering
		\includegraphics[width=10cm]{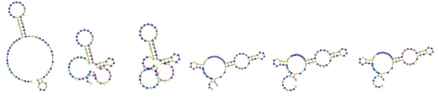}
	\caption{From left to right: RNA secondary suboptimal structures within different bacteria.}
\label{fig_RNA}
\end{figure}

Persistent homology appears as a fundamental tool in Topological Data Analysis.
It studies the evolution of $k$-dimensional holes along a sequence $F$ of simplicial complexes.
The persistence barcode $B(F)$ of $F$ is the collection of  bars (i.e. intervals) representing birth and death times of 
$k$-dimensional holes along such sequence.
In $B(F)$,  $k$-dimensional holes with short lifetimes are informally considered to be ``topological noise'', and those with long lifetimes are  ``topological features'' of the given data.

Persistent homology based techniques are nowadays widely used for analyzing high dimensional data-set and they are good tools for shaping these data-set and for understanding the meaning of the shapes.  Persistent homology reveals the global structure of a data-set and it is a powerful tool for dealing with high dimensional data-set without performing dimensionality reduction.  There are several techniques for building a topological space from the data. The main approach is to complete the data to a collection of combinatorial objects, i.e. simplices. A nested collection of simplices forms a simplicial complex. Simplicial complexes can be obtained from graphs and point cloud data (PCD)~\cite{binchi2014jholes,adams2011javaplex}. 
For example,  PCD can be completed to simplicial complexes by using the Vietoris-Rips  filtration,
which is  a sequence of simplicial complexes  built on a metric space, providing in this way a topological structure to an otherwise disconnected set of points. It is widely used in TDA because it encodes useful information about the topology of the underlying metric space. 
Let us take a look at Fig.~\ref{fig_RNA}, it represents a collection of RNA secondary sub-optimal structures within different bacteria. All the shapes are characterized by several circular substructures, each of them is obtained by linking different nucleotides. Each substructure encodes functional properties of the bacteria. 
Mamuye et al. \cite{Adane} used Vietoris-Rips complexes and persistent homology for certifying that there are different species but characterized with the same RNA sub-optimal secondary structure, thus these species are functionally equivalent. 
The mathematical details of Vietoris-Rips filtration are given in Section \ref{Topology} of this paper. 

Nevertheless, Vietoris-Rips based analysis suffers of the selection of the parameter $\epsilon$. 
Generally speaking, for different $\epsilon$, different topological features can be observed. 
For example, in~\cite{jonoska2013discrete}, several applications of Vietoris-Rips based analysis to biological problems have been reported and examples of different $\epsilon$ with different meaning were found. In order to select the best $\epsilon$, some statistics have been provided what  is  known as `` persistence landscape''~\cite{bubenik2015statistical}. Landscape is a powerful tool for statistically assessing the global shape of the data over different $\epsilon$. Technically speaking, a landscape is a piecewise linear function that basically maps a point within a persistent diagram (or barcode) to a point in which the $x$-coordinate is the average parameter value over which the feature exists, and the $y$-coordinate is the half-life of the feature. Landscape analysis allows to identify topological features. 
In Section~\ref{entropy}, we present the notion of persistent entropy, as an alternative approach to landscape. The main difference between landscape and our method is that the former uses the average of $\epsilon$, while the latter works directly on fixed $\epsilon$.
More concretely, persistent entropy (which is the Shannon entropy of the persistence barcode) is a tool formally  defined in  \cite{rucco2015characterisation} and used to measure similarities between two persistence barcodes.
A precursor of this definition was given in \cite{chintakunta2015entropy} to measure how different the bars of a barcode are in length.
In \cite{signal}, persistent entropy is used for addressing the comparison between discrete piecewise linear functions.
In Section \ref{sec:entropy_rips}, several properties of the persistent entropy of Vietoris-rips filtrations are presented. For example, the exact formula of maximum and minimum persistent entropy is given for a persistence barcode, fixing the number of bars and the maximum and minimum length. These results are important later in Section \ref{method} for differentiate topological features from noise.  

In  general, ``very'' long living bars (long lifetime) are considered topological features since they are stable to ``small'' changes in the filtration. 
In \cite{confidence} a methodology is presented
for deriving confidence sets for persistence diagrams  to separate topological noise from topological features. 
The authors focused on simple, synthetic examples as proof
of concept. 
Their
methods have a simple visualization: one only needs to add a band around the
diagonal of the persistence diagram. Points in the band are consistent with
being noise.
The first three methods in that paper were based on the
distance function to the data. 
They started with a sample from a distribution $\mathbb{P}$ supported on a topological space $\mathfrak{C}$. The bottleneck distance was used as a metric on the space of persistence diagrams. 
The last method in that paper used density estimation. The advantage of the former was
that it is more directly connected to the raw data. The advantage of the
latter was that it is less fragile; that is, it is more robust to noise and outliers.
In Section \ref{method} in this paper,
we derive a simple method for separating topological  features from noise of  a
given filtration  using the mentioned persistent entropy measurement. 
Moreover, we claim it is very easy (and fast) to compute, and easy to adapt depending on the application. A preliminary version of this technique was also presented in \cite{datamod}.

\section{Background}
\label{Topology}

This section provides a short recapitulation of the basic concepts  needed as a basis for the presented method for separating topological noise from features.

Informally, a topological space is a set of points each of them equipped with the notion of neighboring. 
A {\it simplicial complex} is a kind of topological space  constructed by the union of $k$-dimensional simple pieces in such a way that the common intersection of two pieces are lower-dimensional pieces of the same kind. More concretely,
 $K$ is composed by a set $K_0$ of {\it $0$-simplices} (also called vertices $V$,  that can be thought as points in $\mathbb{R}^d$);
and, for each $k\geq 1$, a set $K_k$ of {\it $k$-simplices} $\sigma=\{v_0, v_1, \dots, v_k\}$, where $v_i \in V$ for all $i\in \{0,\dots,k\}$, satisfying that:
\begin{itemize}
\item each $k$-simplex has $k+1$ {\it faces} obtained by removing one of its vertices;
\item if a simplex  is in $K$, then all its faces must be in $K$.
\end{itemize}
The underlying topological space of $K$ is the union of the geometric realization of its simplices: points for $0$-simplices, line segments for $1$-simplices, filled triangles for $2$-simplices,  filled tetrahedra for $3$-simplices and their $k$-dimensional counterparts for $k$-simplices.
 We only consider finite  simplicial complexes with finite dimension, i.e., there exists an integer $m$ (called the dimension of $K$) such that for $k>m$, $K_k=\emptyset$ and, for $0\leq k\leq m$, $K_k$ is a finite set.

Two classical examples of simplicial complexes are $\cech$ech complexes and  Vietoris-Rips complexes (see \cite[Chapter III]{edels}). 
Let $V$ be a (finite) PCD in $\mathbb{R}^d$.  The {\it $\cech$ech complex} of $V$ and $r$ denoted by {\it Č}$_{V}(r)$ is  the  simplicial complex whose simplices are formed as follows. For each subset $S$ of points in $V$, form   a  closed ball of radius $r$ around each point in $S$, and include $S$ as a simplex of {\it Č}$_{V}(r)$  if there is a common point contained in all of the balls in $S$. This structure satisfies the definition of abstract simplicial complex.
The {\it Vietoris-Rips complex} denoted as $\VR_{V}(r)$ is essentially the same as the $\cech$ech complex. Instead of checking if there is a common point contained in the  intersection of  the $(r)$-ball around $v$ for all $v$ in $S$, we may
just check pairs
adding $S$ as a simplex of {\it Č}$_{V}(r)$ if   all the balls have pairwise intersections.  We have 
 $\mbox{\it Č}_{V}(r) \subseteq \VR_{V}(r)\subseteq \mbox{\it Č}_{V}(\sqrt{2} r)$. See Fig.\ref{figure_cech}.
In practice,  Vietoris-Rips complexess are more often used since they are easier to compute than $\cech$ech omplexes.

\begin{figure}[t!]
	\centering
		\includegraphics[width=6cm]{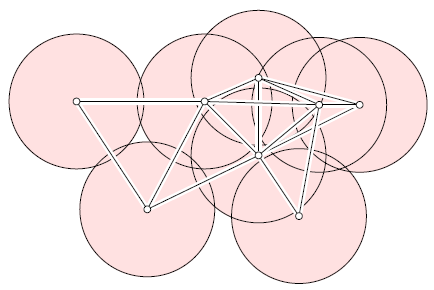}
	\caption{\cite[p. 72]{edels} Nine points with pairwise intersections among the disks indicated by
straight edges connecting their centers, for a fixed time $\epsilon$. The 
$\cech$ech complex $\cech_{V}(\epsilon)$ fills nine of the ten possible
triangles as well as the two tetrahedra. The Vietoris-Rips complex
$\VR_{V}(\epsilon)$
fills the ten 
triangles and the two tetrahedra. }
	\label{figure_cech}
\end{figure}

Homology is an algebraic machinery used for describing topological spaces. The {\it $k$-Betti number} $\beta_k(K)$ represents  the rank of the $k$-dimensional homology group $H_k(K)$ of a given simplicial complex $K$.
Informally, $\beta_0(K)$ is the number of connected components of $K$, $\beta_1(K)$  counts the number of tunnels, $\beta_2(K)$ can be thought as the number of voids of $K$ and, in general,  $\beta_k(K)$ can be thought as the number of $k$-dimensional holes of $K$.
More precisely, homology groups are defined from an algebraic structure called chain complex composed by a set of groups $\{C_k(K)\}_k$, where each $C_k(K)$ is the group of $k$-chains generated by all the $k$-simplices of $K$, and a set of homomorphisms $\{\partial_k: C_k(K)\to C_{k-1}(K)\}_k$, called boundary operators, describing the boundaries of $k$-chains.
A $k$-chain $a$ such that $\partial_k(a)=0$
is a {\it $k$-cycle}.
It is a {\it $k$-boundary} if there exists a $(k+1)$-chain $b$ such that $\partial_{k+1}(b)=a$.
This way, the {\it $k$-dimensional homology group $H_k(K)$} is the group of $k$-cycles modulo the group of  
$k$-boundaries. 
The {\it $k$-th Betti number} $\beta_k(K)$ is the rank of $H_k(K)$.
See \cite{hatcher2002algebraic} and \cite{munkres1984elements} for an introduction to algebraic topology.

Persistent homology is a method for computing $k$-dimensional holes of a given topological space at different spatial resolutions. 
The key idea is as follows.

\begin{itemize}
\item  First, 
the space must be represented as a simplicial complex $K$ and a distance function  must be defined on the space. 
\item 
Second, a {\it filtration} of $K$, referred above as different spatial resolutions, is computed. 
More concretely, 
a filtration $F$ of  $K$   is a collection of simplicial complexes
F=$\{K(t)\,|\, t \in \mathbb{R}\}$ of $K$ such that $K(t) \subset K_s$ for $t< s$ and there exists $t_{\max}\in \mathbb{R}$ such that $K_{t_{\max}}=K$. The filtration time (or filter value) of a simplex $\sigma \in K$ is the smallest $t$ such that $\sigma \in K(t)$.
For example,
let $V$ be a PCD in $\mathbb{R}^d$ and let $r,r'\in {\mathbb R }$.
Then, there is a natural inclusions {\it Č}$_V(r)\subseteq ${\it Č}$_{V}(r')$ and $\VR_V(r)\subseteq \VR_V(r')$ whenever $r\leq r'$.  
The simplicial complexes {\it Č}$_V(r)$ together with the inclusion maps define a filtered simplicial complex {\it Č}$_V$ called {\it $\cech$ech filtration}.
Similarly, the simplicial complexes $\VR_V(r)$ together with the inclusion maps define a filtered simplicial complex $\VR_V$ called {\it Vietoris-Rips filtration}. 
\item
Then, persistent homology describes how the homology of a given simplicial complex $K$ changes along filtration
$F=\{K(t) \,|\, t \in \mathbb{R}\}$. 
If the same topological feature (i.e., $k$-dimensional hole) is detected along a large number of subsets in the filtration, then it is 
likely to represent a true feature of the underlying space, rather than artifacts of sampling, noise, or particular choice of parameters.
More concretely, a bar in the $k$-dimensional persistence barcode, with endpoints $[t_{start}, t_{end}),$ corresponds to a $k$-dimensional hole that appears at filtration time $t_{start}$ and remains until  filtration  time $t_{end}$. 
The set of bars $[t_{start}, t_{end})$ representing birth and death times of homology classes is called the {\it persistence barcode} $B(F)$ of the filtration $F$.
Analogously, the set of points $(t_{start}, t_{end})\in {\mathbb R}^2$  is called the {\it persistence diagram} $dgm(F)$ of the filtration $F$.
\end{itemize}
For more details and a more formal description  we refer to~\cite{edels}.

Classically, the bottleneck distance (see \cite[page 229]{edels}) is used to compare the persistence diagrams of two different filtrations. 
Concretely, 
let $dgm(F)=\{a_1,\dots,a_k\}$ and $dgm(F')=\{a'_1,\dots,a'_{k'}\}$ be, respectively, the persistence diagram $dgm(F)$ and $dgm(F')$ of the  two filtrations $F$ and $F'$,  then 
$$d_b(dgm(F),dgm(F'))=\inf_{\gamma}\{\sup_v\{||v-\gamma(v)||_{\infty}\}\}$$ 
is the {\it bottleneck distance} between $dgm(F)$ and $dgm(F')$
where,
for points $a=(x,y)$ and $\gamma(a)=(x',y')$ in $\mathbb{R}^2$, 
$||a-\gamma(a)||_{\infty}=\max\{|x-x'|,|y-y'|\}$ and 
$\gamma: dgm(F)\to dgm(F')$ is a bijection that can associate a point off the diagonal with another point on or off the diagonal. Here, {\it diagonal} is the set of points $\{(x,x)\}\subset {\mathbb R}^2$.

\begin{remark}\label{uno} Since simplicial complexes considered in this paper are finite then for given filtrations $F$ and $F'$,
we have that:
\begin{itemize}
\item  $dgm(F)$ is a finite set of points in ${\mathbb R}^2$.
\item  $d_b(dgm(F),dgm(F'))=\min_{\gamma}\{\max_a\{||a-\gamma(a)||_{\infty}\}\}$.
\end{itemize} 
\end{remark}

In the following theorem, it is state that low-distortion correspondences between two PCDs,  $V$ and $W$,  in $\mathbb{R}^d$ give rise to small distance in the bottleneck distance of the persistence diagrams of the  $\cech$ech filtrations $\mbox{\it Č}_V$ and $\cech_W$ and the Vietoris-Rips filtrations $\VR_V$ and $\VR_W$. 
  
\begin{theorem}{\bf Persistence stability for $\cech$ech and Vietoris-Rips  complexes} \cite[Th. 5.2.]{chazal}\label{th:chazal}
Let $V$ and $W$ be two sets of points in $\mathbb{R}^d$ then, 
for either $F_V=\mbox{\it Č}_V$ and $F_W=\mbox{\it Č}_W$ or $F_V=\VR_V$ and $F_W=\VR_W$, we have that:
$$d_b(dgm(F_V),dgm(F_W))\leq 2 d_{GH}(V,W),$$
where $2d_{GH}(V,W)=\inf_c\{sup_{v,v'} |d(v,v')-d(c(v),c(v'))|\} \}$
for $c:V\to W$ being surjective. 
\end{theorem}

\begin{remark}\label{dos} 
Since PCDs considered in this paper are finite,  then
 $$2d_{GH}(V,W)=\min_c\{\max_{p,p'} |d(p,p')-d(c(p),c(p'))|\} \}.$$
\end{remark}


\section{Persistent entropy}\label{entropy}

In order to measure how much the construction of a filtration is ordered, a new entropy measure, the so-called \textit{persistent entropy}, were defined in~\cite{rucco2015characterisation}. A precursor of this definition was given in~\cite{chintakunta2015entropy} to measure how different the bars of a barcode were in length. In \cite{signal},  persistent entropy was used  for addressing the comparison between discrete
piece-wise linear functions.

\begin{definition} Given a filtration $F=\{K(t) \,|\,  t\in \mathbb{R}\}$ and the corresponding persistence diagram $dgm(F) = \{a_i=(x_i , y_i) \,|\,  1\leq i\leq n\}$ 
(being $x_i< y_i$ for all $i$), let $L=\{\ell_i=y_i-x_i \,|\, 1\leq i\leq n\}$. The \textit{persistent entropy} $E(F)$ of  $F$ is calculated as follows:
$$
E(F)=-\sum_{i =1}^n p_i  \log(p_i)
\mbox{\;
where $p_i=\frac{\ell_i}{S_L}$, $\,\ell_i=y_i - x_i$, and }S_L=\sum_{i=1}^n \ell_i.
$$
Sometimes, persistent entropy $E(F)$ will  also be denoted by $E(L)$.
\end{definition}

Note that the maximum persistent entropy would correspond to the situation in which all the bars in the associated persistence barcode are of equal length (i.e., $\ell_i=\ell_j$ for all $1\leq i, j\leq n$). Conversely, the value of the persistent entropy decreases as more bars of different lengths are present in the persistence barcode.  
More concretely,  if $E(F)$ has $n$ points,  the possible values of  $E(F)$ lie in the interval $[0,\log(n)]$.

The following result supports the idea that persistent  entropy can differentiate long from short bars as we will see in Section \ref{method}.
\begin{theorem}\label{theorem}\cite{datamod}
Given a filtration $F$ and the corresponding persistence diagram $dgm(F) = \{a_i=(x_i , y_i) \,|\,  1\leq i\leq n\}$, let 
$L=\{\ell_i=y_i-x_i \,|\, 1\leq i\leq n\}$.
For a fixed integer $i$, $1\leq i\leq n$, 
let  
$$L'=\{\ell'_1,\dots,\ell'_{i},\ell_{i+1},\dots,\ell_n\}$$
where $ \ell'_j=\frac{P_{i}}{e^{E(R_{i})}}$ for $1\leq j\leq i$,
$R_{i}=\{\ell_{i+1},\dots\ell_n\}$ and  $P_{i}=\sum_{j=i+1}^{n}\ell_j$. 
  Then 
 $$E(L)\leq E(L').$$
 \end{theorem}
 
Observe that we can also write  
$ \ell'_j=\prod_{j=i+1}^n\ell_j^{\ell_j/P_{i}}$. 
This last expression will be very useful 
in the proof of Th.  \ref{th:parada} in Section \ref{method}.

\proof
Let us prove that $E(L')$ is the maximum of all the possible persistent entropies associated to  barcodes with $n$ bars, such that the list of lengths of the last $n-i$ bars of any of such lists is $R_{i}$. 
Let $M= \{x_1,\dots,x_{i},\ell_{i+1},\dots,\ell_n\}$ (where $x_j>0$ for 
$1\leq j\leq i$) be any of such lists.
\\
 Let  $S_x=\sum_{j=1}^{i}x_j$. 
 Then,  the persistent entropy associated to $M$ is:
\begin{eqnarray*}
&&E(M)
=-\sum_{j=1}^{i}
\frac{x_j}{S_x+P_{i}}\log\left(\frac{x_j}{S_x+P_{i}} \right)
-\sum_{j=i+1}^n\frac{\ell_j}{S_x+P_{i}}\log\left(\frac{\ell_j}{S_x+P_{i}} \right)\\
&&=-\sum_{j=1}^{i}
\frac{x_j}{S_x+P_{i}}\log\left(\frac{x_j}{S_x+P_{i}} \right)
-\frac{P_{i}E(R_{i})}{S_x+P_{i}}-\frac{P_{i}}{S_x+P_{i}}\log \left(\frac{P_{i}}{S_x+P_{i}}\right). 
\end{eqnarray*}   
In order to find out the maximum of $E(M)$ with respect to the unknown variables 
$x_k$, $1\leq k\leq i$, we compute the partial derivative of $E(M)$ 
with respect to those variables:
\begin{eqnarray*}
\frac{\partial E(M)}{\partial x_k}
=\frac{1}{(S_x+P_{i})^2} &&\left(P_{i}E(R_{i})+
P_{i}\log\left(\frac{P_{i}}{x_k}\right)
+
\sum_{j\neq k}x_j\log\left(\frac{x_j}{x_k}\right)\right).
\end{eqnarray*}
Finally, $\left\{x_k=\frac{P_{i}}{e^{E(R_{i})}}\,|\, 1 \leq k\leq i\right\}$ is the solution of  $\left\{\frac{\partial E(M)}{\partial x_k}=0 \,|\, 1\leq k \leq i\right\}$.
\qed

The following result establishes a relation between bottleneck distance and persistent entropy.
\begin{proposition}\label{prop:bot}
Let 	$F$ and $F'$ be two filtrations. 
For all  $\epsilon>0$, there exists $\delta>0$ such that if  $d_b(dgm(F),dgm(F'))<\delta$ then $|E(F)-E(F')|<\epsilon$.
  \end{proposition}  

\begin{proof} 
The proof is similar to the one given in \cite{signal} to demonstrate that persistent entropy associated to  piece-wise linear functions is stable.
\\
Fixed $\epsilon>0$, we have to find $\delta>0$ such that if  $d_b(dgm(F),dgm(F'))<\delta$ then $|E(F)-E(F')|<\epsilon$.
\\
First, since $h(x)=-x\log x$ is a continuous function in $[0,1]$ (redefining $h(0)$ as $0$), for $\epsilon'=\frac{\epsilon}{n}>0$, there exists $\delta'\in (0,1]$ such that if $|x-x'|\leq \delta'$ then
$|h(x)-h'(x)|\leq \epsilon'$.
\\
Take  $\delta=\frac{S_{L'}\delta'}{4n}$ and suppose  $d_b(dgm(F),dgm(F'))<\delta$.
\\
By Remark \ref{uno}, $dgm(F)$ and $ dgm(F')$ are both finite and there  exists a bijection $\bar{\gamma}: dgm(F)\to dgm(F')$ such that  
$d_b(dgm(F),dgm(F'))=\max_a\{||a-\bar{\gamma}(a)||_{\infty}\}$.
Let  $dgm(F)=\{a_1,\dots,a_n\}$   (where some of the $a_i$  can possibly be on the diagonal). 
Let $a_i=(x_i,y_i)$ and $\bar{\gamma}(a_i)=(x'_i,y'_i)$.
Then,  
$$\mbox{$||a_i-\bar{\gamma}(a_i)||_{\infty}=\max\{|x_i-x'_i|,|y_i-y'_i|\}\leq \delta$ for all $i$.}$$
Let $\ell_i=y_i-x_i$ and $\ell'_i=y'_i-x'_i$. Then, 
$$\mbox{$|\ell_i-\ell'_i|=|x_i-y_i-(x'_i-y'_i)|\leq |x_i-x'_i|+|y_i-y'_i| \leq 2\delta$ for all $i$.}$$
Besides,  $$|S_{L}-S_{L'}|=\left|\sum_{i=1}^n \ell_i-\sum_{i=1}^n \ell'_i  \right|\leq  \sum_{i=1}^n |\ell_i-\ell'_i|\leq 2\delta n.$$
\\
Without lost of generality, assume $S_{L}\geq S_{L'}$. Then $S_{L}\leq S_{L'}+2\delta n$.  
\\
Let $p_i=\frac{\ell_i}{S_{L}}$ and $p'_i=\frac{\ell'_i}{S_{L'}}$. Then
$$p_i-p'_i=\frac{\ell_i}{S_{L}}-\frac{\ell'_i}{S_{L'}}=\frac{S_{L'}\ell_i-S_{L}\ell'_i}{S_{L}S_{L'}}\leq \frac{\ell_i-\ell'_i}{S_{L'}}\leq \frac{2\delta}{S_{L'}}= \frac{\delta'}{2n}\leq \delta';$$
$$p'_i-p_i\leq \frac{(S_{L'}+2\delta n)\ell'_i-S_{L'}\ell_i}{S_{L}S_{L'}}
\leq \frac{\ell'_i-\ell_i}{S_{L'}}+\frac{2\delta n\ell'_i}{S_{L'}S_{L'}}
\leq
\frac{2\delta n}{S_{L'}}\left(1+\frac{\ell'_i}{S_{L'}}\right)\leq \delta'.$$
Therefore, 
$$|E(F)-E(F')|=
\left|\sum_{i=1}^n p_i\log p_i-\sum_{i=1}^n p'_i\log p'_i\right|\leq
\sum_{i=1}^n |p_i\log p_i-p'_i\log p'_i|\leq 
\epsilon,$$
which concludes the proof.
\qed
\end{proof}
 The  result above is  used now to prove that  persistent entropy is a stable measure for $\cech$ech and Vietoris-Rips filtrations.

\begin{theorem}
{\bf Persistent entropy stability theorem for $\cech$ech and Vietoris-Rips filtrations.}
Let $V$ and $W$ be two PCDs in $\mathbb{R}^d$. Then, for every $\epsilon>0$ there exists $\delta>0$ such that:
$$\mbox{If $2d_{GH}(V,W)\leq \delta$ then $|E(F_V)-E(F_W)|<\epsilon$,}$$
where either  $F_V=\mbox{\it Č}_V$ and $F_W=\mbox{\it Č}_W$ or
$F_V=\VR_V$ and $F_W=\VR_W$.
\end{theorem}

\begin{proof}
First, by Prop. \ref{prop:bot} we have that fixed $\epsilon>0$,
there exists $\delta>0$ such that 
if $d_b(dgm(F_V), dgm(F_W)) <\delta$ then $|E(F_V) - E(F_W)| < \epsilon$.
\\
Second, by Th. \ref{th:chazal} we have that
$d_b(dgm(F_V),dgm(F_W))\leq 2 d_{GH}(V,W).$
\\
Therefore, if $2 d_{GH}(V,W)<\delta$ then $|E(F_V) - E(F_W)| < \epsilon$.
\qed
\end{proof}

\section{Properties of the persistent entropy  of  Vietoris-Rips
filtrations}\label{sec:entropy_rips}

Since Vietoris-Rips filtration  are widely used in practice, 
we focus now our effort in the study of properties
of the persistent entropy  of  this special kind of 
filtrations.

The first thing we have to take into account is that, in practice, one will never construct the filtration up to the end and will stop at a certain time $T$. Then, $\VR_V=\{\VR_V(t) \,|\,  t\leq T\}$. 
To decide when to stop, we use the following result.

\begin{proposition}\label{stop}
Let $V=\{v_1,\dots,v_m\}$ be a PCD in  ${\mathbb R}^d$. Let  $$T=\frac{\min_i\max_j d(v_i,v_j) }{2}.$$ Then, $\beta_0(\VR_V(T))=1$ and $\beta_k(\VR_V(T))=0$ for $k>0$. 
\end{proposition}

\begin{proof}
First, notice that there exists a vertex $v$ such that $\max_j d(v,v_j)=2T$. That is, 
$d(v,v_j)\leq 2T$ for  $1\leq j\leq m$. Then, $v$ is connected to $v_j$ by an edge
in $\VR_V(T)$, for 
 $1\leq j\leq m$. In particular, $\beta_0(\VR_V(T))=1$.
 \\
 Now,  
 observe that if
 $\sigma=\{v_0, v_1, \dots, v_k\}$ is a $k$-simplex in $\VR_V(T)$ and
 $v\not\in \sigma$, then $\sigma\cup \{v\}=\{v,v_0, v_1, \dots, v_k\}$ is a $(k+1)$-simplex in  $\VR_V(T)$ and $\partial_{k+1}(\sigma\cup \{v\})=
\sigma+\partial_{k}(\sigma)\cup \{v\} $.
 \\
Let $c=\sum_{i\in I} \sigma_i$ be a   cycle in $C_k(\VR_V(T))$. 
Let $J=\{j \,|\, j\in I$ and $v$ is not a vertex of $\sigma_j\}$. Let $b=\sum_{j\in J}\sigma_j\cup\{v\}$. Then 
 $$\partial_{k+1}(b)=\sum_{j\in J}\sigma_j+\partial_{k}(\sigma_j)\cup\{v\}=
\sum_{j\in J}\sigma_j+\sum_{i\in I\setminus J}\sigma_i=c.$$
 Therefore, $c$ is a boundary. Then $\beta_k(\VR_V(T))=0$ for $k>0$.
 \qed
\end{proof}

From now on, given a PCD $V=\{v_1,\dots,v_m\}$ in $\mathbb{R}^d$,  we   construct $$\VR_V=\{\VR_V(t)\,|\,t\leq T\}\;\mbox{ for $T=\frac{\min_i\max_j d(v_i,v_j) }{2}$.}$$
By Prop. \ref{stop}, the biggest bar in the persistence barcode in dimension $0$  was born at time $t=0$ and survives until the end (i.e., time $t=T$) and the smallest bar was born at time $t=0$ and survives until $t=r=\min_{i,j}d(v_i,v_j)$.
Fixed the number of bars in the persistence barcode and the maximum  and minimum lengths of the bars, $T$ and $r$, the following result shows the lengths of the rest of the bars that provide the minimum persistent entropy. This result will be very useful in the next section to detect topological features.  

\begin{theorem}\label{th:min} Let $L=\{\ell_1,\dots,\ell_n\}$ such that 
$\ell_1=T$, $\ell_i\geq \ell_{i+1}$ for $1\leq i<n$ and $\ell_n=r$.
Let $M=\{T,\stackrel{Q}{\dots},T,r,\stackrel{n-Q}{\dots},r\}$.
Then 
$$E(L)\geq E(M)\;\mbox{ for $Q=\left[\frac{\alpha n(\alpha-1-\log(\alpha))}{(\alpha-1)^2} \right]$ being $\alpha=\frac{r}{T}$.}$$
\end{theorem}

 \begin{proof}
 First, fixed $n$, $T$ and $r$, 
 Let $p_i=\frac{l_i}{S_L}$. Since the entropy is a concave function in 
 $$\Omega=\left\{(p_1,p_2,\dots p_n)\,|\,\sum\limits_{i=1}^n p_i=1,\frac{1}{(n-1)+\alpha}<p_i<p_1<\frac{1}{(n-1)\alpha+1}\right\},$$
 being $\alpha=\frac{r}{T}$, the minimum is attained at an extremal point of $\Omega$. Let 
 $$\mbox{$P=(p_1,\stackrel{i}{\dots},p_1,\alpha p_1,\stackrel{n-i}{\dots},\alpha p_1)$, with $1<i<n$,}$$
 be an extremal point. Since $\sum\limits_{i=1}^n p_i=1$, then $p_1=\frac{1}{i+\alpha(n-i)}$ and the entropy of $P$ is:
 $$E(P)= \log(\alpha n)+ \log(1+\frac{i(1-\alpha)}{\alpha n})-\frac{\log(\frac{1}{\alpha})}{1+(\frac{n}{i}-1)\alpha}.
 $$
 Consider $t=\frac{i}{n}\in(0,1)$,
 then:
 $$E(P)=E(t)=\frac{\alpha(1-t)\log(\frac
 {1}{\alpha})}{\alpha(1-t)+t}$$
 The derivative of $E(t)$ is null for:
 $$t_0=\frac{\alpha\log(\frac{1}{\alpha})-\alpha(1-\alpha)}{(1-\alpha)^2}$$
 So the minimum entropy is attained for 
 $$Q=
\left[ n t_0\right]=\left[n \frac{\alpha\log(\frac{1}{\alpha})-\alpha(1-\alpha)}{(1-\alpha)^2}\right].$$
Taking in account that $p_1=T$, the barcode with $l_1=T$ and $l_n=r$ with minimum entropy is $M=\{T,\stackrel{Q}{\dots},T,r,\stackrel{n-Q}{\dots},r\}$. 
\qed
 \end{proof}
 
In the following proposition, we establish the maximum entropy we can reach for $n$ bars fixing the the maximum  and minimum lengths of the bars.

 \begin{proposition}\label{prop:max}
 Fixed $n$, $T$ and $r$, Let $L=\{\ell_1,\dots,\ell_n\}$ such that 
 $\ell_1=T$, $\ell_i\geq \ell_{i+1}$ for $1\leq i<n$ and $\ell_n=r$.
 Let $$M'=\{T, b,\stackrel{n-2}{\dots},b,r\},\;\mbox{
 where $b=T\alpha^{\alpha/(1+\alpha)}$ and $\alpha=\frac{r}{T}$. }$$
 Then $E(L)\leq E(M')$.
 \end{proposition}
 \begin{proof}
  First, reorder the list $L=\{\ell_2,\dots,\ell_n,\ell_1\}$ and then neutralize the bars $\ell_j$ for $2\leq j\leq n-1$. By Th. \ref{theorem}, the new values that provide the maximum entropy are: $$\ell'_j=T^{T/(r+T)}r^{r/(r+T)}=T^{1/(\alpha+1)}r^{\alpha/(\alpha+1)}
  =T\alpha^{\alpha/(\alpha+1)}.$$
\qed
 \end{proof}

Th. \ref{th:min} and Prop. \ref{prop:max} confirm that the possible values of the persistent entropy $B(F)$ of a filtration $F$ associated to a PCD $V$ is highly influenced by the number $n$ of bars in $B(F)$ and the rate between the minimum and maximum persistence entropy that we can reach with  $n$ bars. This rate is also influenced by the minimum distance $r$ between two points in the PCD and the radius $2T$ of $V$.
Now, given a persistence barcode $L$ with $n$ bars, maximum length of the bars equal to $T$ and minimum length equal to $r$, the relative entropy $E(L)/E(M')$, being $M'$ the possible maximum entropy with same data $n$, $r$ and $T$, allows us  to compare 
two persistence barcodes  with different numbers of bars. Finally, observe that the value of $Q$ in Th. \ref{th:min} gives us a quantity of the maximum number of topological features we can find fixing the length of the persistence barcode and the maximum and minimum length of the bars. 

\section{Separating topological features from noise}
\label{method}

Let us start now with a PCD $V=\{v_1,\dots,v_m\}$ in $\mathbb{R}^d$ from a distribution $\mathbb{P}$ supported on
a topological  space $\mathfrak{C}$. 
Suppose the  Vietoris-Rips filtration $\VR_V$ is computed from $V$ (being $T=\frac{\min_i\max_jd(v_i,v_j)}{2}$), and the persistence barcode $B(\VR_V)$  is computed from $\VR_V$. 
The following are the steps of our proposed method, based on persistent entropy, to separate topological noise from topological features in the persistence barcode $B(\VR_V)$, estimating, in this way, the topology of $\mathfrak{C}$.

\begin{alg}\label{proc} Computing topological features from the persistent barcode  of the Vietoris-Rips filtration   of a given PCD  in $\mathbb{R}^d$.
\begin{itemize}
\item[] {\bf Input:} A PCD $V=\{v_1,\dots,v_m\}$ in $\mathbb{R}^d$, its Vietoris-Rips filtration  $\VR_V=\{\VR_V(t)\,|\, t\leq T\}$ and  its associated persistence barcode $B(\VR_V)=\{[x_i,y_i)$ $\,|\, 1\leq i\leq n\}$.
\item[1.]  Sort the lengths of the bars in $B(\VR_V)$
in decreasing order, except for the longest bar (whose length is equal to  $T$) to obtain 
$L=\{\ell_1,\dots,\ell_n\}$ such that $\ell_n=T\geq \ell_1$ and $\ell_i\leq \ell_j\leq \ell_{n-1}=r=\min_{i,j}d(v_i,v_j)$ for $1\leq i<j<n-1$. \\
Initially,  $L'_0:=L$ and $n':=n$.
\item[2.] {\bf For} $i=1$ to $i=n'-2$ {\bf do}:
\begin{itemize}
\item[a.] Compute the persistent entropy $E(L'_i)$ for 
$L'_i=\{\ell'_1,\dots,\ell'_{i},\ell_{i+1},\dots,$ $\ell_{n'}\}$, being
$\ell'_j=\frac{P'_{i}}{e^{E(R'_{i})}}$ for $1\leq j\leq i$ as in Th. \ref{theorem}.
\item[b.] Compute $$
C=\frac{S_{L'_{i-1}}}{S_{L'_i}}=
\frac{P'_{i-1}+(i-1)\frac{P'_{i-1}}{e^{E(R'_{i-1})}}}{P'_{i}+i\frac{P'_{i}}{e^{E(R'_{i})}}}\;\mbox{ and }\; Q=\left[\frac{\alpha n'(\alpha-1-\log(\alpha))}{(\alpha-1)^2}\right]
$$
being $\alpha=\frac{r}{T}$.
\end{itemize}
{\bf while} $C\geq 1$.
\item[3.] {\bf If} $Q<i$, then the bars $[x_j,y_j)$ with $i<j<n-1$ represent noise in the pesistence barcode. Redefine $L_0':=L'_0\setminus\{\ell_{i+1},\dots,\ell_{n'-2}\}$ and $n':= i+2 $. 
{\bf Go to step 2.} 
\\
{\bf Else}, the  bars of $B(VR_V)$ with lengths in the set $\{T,\ell_1,\dots, \ell_i\}$ represent topological features of $\VR_V$.
\item {\bf Output:} The bars of $B(VR_V)$ that represent topological features  of $\VR_V$.
\end{itemize}
\end{alg}

The following result guarantees the end of the while-loop in Proc. \ref{proc}.

\begin{theorem}\label{th:parada}
Fixed $n'$ in Proc. \ref{proc}, there always exists a value  $i$, $1\leq i\leq n'-2$, such that 
$C=\frac{S_{L'_{i-1}}}{S_{L'_i}}<1$ except when  $T=r$ (which corresponds to a uniform distribution and, in this case, $Q=n'$).
\end{theorem}
\begin{proof}
Observe that $L'_{n'-2}=\{b,\stackrel{n'-2}\cdots b,r,T\}$ for $b$ as in Prop. \ref{prop:max}.\\
First, 
if $S_{L'_0}<S_{L'_{n'-2}}$ then:
$$\frac{S_{L'_0}}{S_{L'_{n'-2}}}=\frac{S_{L'_0}}{S_{L'_1}}
\cdots \frac{S_{L'_{i-1}}}{S_{L'_{i}}}\cdots\frac{S_{L'_{n'-3}}}{S_{L'_{n'-2}}}<1.
$$
Then, there exists $i$, $1\leq i\leq n'-2$, such that 
$\frac{S_{L'_{i-1}}}{S_{L'_{i}}}<1$.\\
Second, if  $S_{L'_0}\geq S_{L'_{n'-2}}$, since  $S_{L'_{n'-2}}=(n'-2)b+r+T$ then there exists $i$, $1\leq i\leq n'-2$, such that $\ell_j\geq b$ for $j\leq i$ and $\ell_j<b$ for $i< j\leq n'-2$. Then, it is enough to prove that $S_{L'_i}<S_{L'_{n'-2}}$. By Th. \ref{theorem}, we have that:
$$S_{L'_i}=\sum_{j=1}^i\ell'_i+P'_i
=i\prod_{j=i+1}^{n'}\ell_j^{\ell_j/P'_i}+P'_{i}
=i\prod_{j=i+1}^{n'-2}\ell_j^{\ell_j/P'_i}T^{T/P'_i}r^{r/P'_i}+P'_{i}.$$
Observe that $P'_i< (n'-i-2)b+r+T$, since $\ell_j<b$ for $i<j\leq n'-2$.
Now, let us prove that 
$$\prod_{j=i+1}^{n-2}\ell_j^{\ell_j/P'_i}T^{T/P'_i}r^{r/P'_i}< b,$$ which,  taking the log of both sides, is equivalent to prove that:
$$\sum_{j=i+1}^{n'-2}\frac{\ell_j}{P'_i}\log(\ell_j)+\frac{T}{P'_i}\log(T)+\frac{r}{P'_i}\log(r)< \frac{T}{T+r}\log(T)+\frac{r}{T+r}\log(r).$$
 Replacing $P'_i$ by 
$\sum_{j=i+1}^{n'-2}\ell_j+r+T$ and simplifying, 
we have to prove that 
$$\sum_{j=i+1}^{n'-2}\ell_j\log(\ell_j)< \frac{T\log(T)+r\log(r)}{T+r}\sum_{j=i+1}^{n'-2}\ell_j.$$
Since, for $i+1\leq j\leq n'-2$,  $\ell_j<b$, then $\log(\ell_j)<\frac{T\log(T)+r\log(r)}{T+r}$ which concludes the proof.
\qed
\end{proof}

Observe that, in Proc. \ref{proc}, for $1\leq i\leq n'$, $E(L'_i)$ is the entropy of the barcode obtained by  replacing the first $i$ bars of $L'_i$ by  $i$ bars that maximize the entropy. Observe that  $E(L'_i)\leq E(L'_j)$ for $1\leq i<j\leq n'$  by Th. \ref{theorem}.
Then, the idea of the algorithm is to successively neutralize bars (using Th. \ref{theorem}) except for the longest and the shortest ones that are intrinsic in the nature of the filtration. We do it until $C=\frac{S_{L'_{i-1}}}{S_{L'_i}}$ is less than $1$. 
What we measure with $C$ is the change of the probability associated to the long bar $[0,T)$ which, in step $i$, is $p_n^{(i)}=\frac{T}{S_{L'_{i}}}$. Observe that if a long bar is neutralize at step $i$, then  
$p_n^{(i-1)}\leq p_n^{(i)}$, since  neutralization in this case means to shorten the bar $\ell_i$ which produces an increase of the probability of the longest bar.
On the other hand, if a short bar is neutralize at step $i$, then  $p_n^{(i-1)}>p_n^{(i)}$, since neutralization in this case means to elongate the bar $\ell_i$. In this last case, all the bars from $\ell_{i+1}$ to $\ell_{n'-2}$ are considered noise and removed from $L'_0$. We remove noise successively until the maximum number of topological features (computed as $Q$ using Th. \ref{th:min}) is reached. Then, the algorithm ends. 

Observe that this method is different from the one presented in \cite{datamod}. In that paper, in order to appreciate the influence of the current length $\ell_i$ in the  initial  persistent entropy $E(L)$, 
we divided $E(L'_i)-E(L'_{i-1})$ by  $\log(n)-E(L)$ to obtain $H_{rel}(i)$. 
Then, we compare $H_{rel}(i)$  with $\frac{i}{n}$ since  $H_{rel(i)}$ is affected by the total number of lengths and the number of lengths we are replacing.
Nevertheless, the  threshold $\frac{i}{n}$ was taken based on experimentation. In this paper, the constants  $C$ and $Q$ and the thresholds $C<1$ are founded on the mathematical results Th. \ref{th:min}, Prop. \ref{prop:max} and Th. \ref{th:parada}.

\begin{figure}[t!]
	\centering
		\includegraphics[width=8cm]{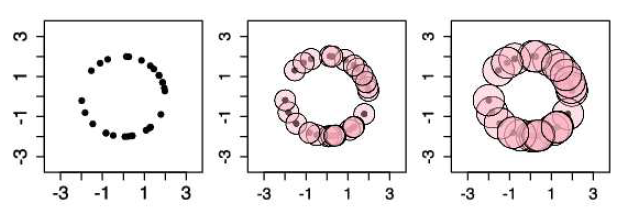}
	\caption{Left: $30$ data points sampled from a circle of radius $1$. Middle: Balls of radius $0.25$ centered at the sample points. 
Right: Balls of radius $0.4$ centered at the sample points.}
	\label{circle}
\end{figure} 

We have applied our methodology to three different scenarios.
First, we take $30$ data points sampled from a circle of radius $1$ (see Fig. \ref{circle}.Left). This example has been taken from paper \cite{confidence}. 
 Vietoris-Rips complex for $t=0.25$ can be deduced from the picture shown in Fig. 
\ref{circle}.Middle which consists of 
 two connected
components and zero loops. 
Looking at Vietoris-Rips complex for $t=0.4$ (see Fig. \ref{circle}.Right), we assist at the birth and death of topological features: at $t= 0.4$, one of the
connected components has died (was merged with the other one), and a loop appears; this loop will die at $t= 1$, when the union of the pink balls representing the distance function becomes simply connected.
In Table \ref{table1}.Top, we have applied our method to the  bars that make up the persistence barcode (without differentiating dimension).  This way, only the bars with length $1$  (that corresponds to the connected component that survives until the end) and $0.6$ (that correspond to the loop that appears at $t=0.4$ and disappears at $t=1$) are considered topological features.
 Later, in Table \ref{table1}.Bottom, we have applied our method to the  bars that make up the persistence $0$-barcode (i.e., the lifetime of the connected components along the filtration).  This way,  the bars with length $1$   and $0.35$ (that corresponds to the connected components that dies just before the loop is created) are considered topological features. 

\begin{table}
\centering
\begin{tabular}{c}
\begin{tabular}{|c|c|c|c|c|}
 \hline
 $ \text{Iteration}$&$n'$& $Q$&  $E(L')/E(M')$ &$\alpha$ \\
 \hline
 1 & 30&5&0.93451& 0.05
 \\
 \hline\hline
   $\ell_i$  &$\ell'_i$ & C& $E(L'_i)/E(M')$ & Feature? \\
 \hline
 1. & 0.0412134&1.06692 & 0.944278 & \text{yes} \\
 0.6 & 0.0409925 &1.0259 &0.945869 & \text{yes} \\
 0.35 & 0.0409923 & 1.00065 &0.945870 & \text{yes} \\
 0.225 & 0.0409921 &1.00071 &0.945872 & \text{yes} \\
0.225 &0.0409835 &0.995049 &0.945934 & \text{yes} \\
0.2 &0.0409741 & 0.99457 &0.946002 & \text{no} \\
0.2 & 0.0409637&0.994018&0.946077 & \text{no} \\ 
\dots &   \dots  &  \dots &  \dots &  \dots\\
\hline
\end{tabular}
\\
\\
\begin{tabular}{|c|c|c|c|c|}
 \hline
$\text{Iteration}$& $n'$& $Q$&   $E(L')/E(M')$ &$\alpha$ \\
 \hline
 2&6&1&0.917626& 0.05\\
 \hline\hline
   $\ell_i$  &$\ell'_i$ & C& $E(L'_i)/E(M')$ & Feature? \\
 \hline
 1. & 0.2165 & 1.03422 &0.918784 & \text{yes} \\
 0.6 & 0.213219 & 0.900062 &0.927951 & \text{yes} \\
 0.35 & 0.20185 & 0.806752 & 0.960854 & \text{no} \\
 0.225 & 0.18911 & 0.746703 & 1. & \text{no} \\
 \dots &   \dots  &  \dots &  \dots &  \dots\\
 \hline
\end{tabular}
\end{tabular}
\caption{Results of our method applied to the barcode obtained from the PCD showed in Fig. \ref{circle} consisting of $30$ data points sampled from a circle of radius $1$.  }
\label{table1}
\end{table}

\begin{figure}[t!]
	\centering
		\includegraphics[width=12cm]{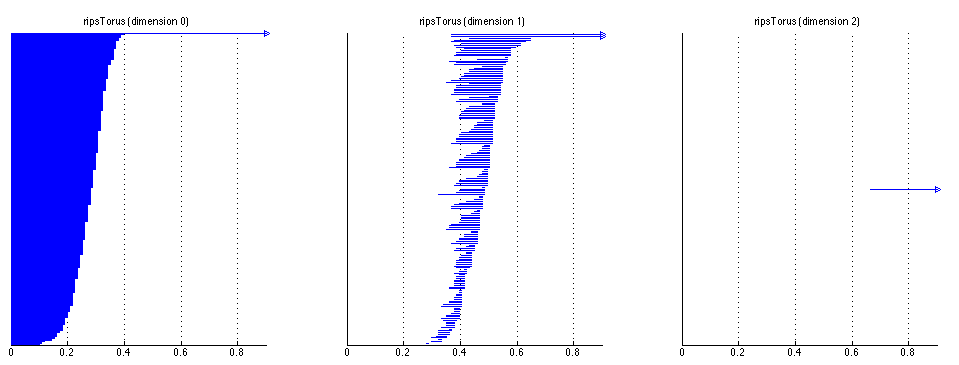}
	\caption{Barcodes (separated by dimension) computed from the Vietoris-Rips filtration associated to a point cloud lying on a 3D torus. Left: lifetimes of connected components. Middle: lifetimes of tunnels. 
Right: lifetimes of voids. }
	\label{figure_torus}
\end{figure}

Second, consider now a set $V$ of $400$ points sampled from  a 3D torus.
The barcodes (separated by dimension) computed from the Vietoris-Rips filtration associated to $V$ are showed in Fig. \ref{figure_torus}. We have applied our method to the $0$-barcode (lifetime of connected components along the V-R filtration) and the $1$-barcode (lifetime of loops  along the V-R filtration).
See Table \ref{table3}. The  bar of length $1.9$ in the tables  corresponds to the connected component that survives until the end. The bars of length $1.531$ correspond to the two tunnels of the 3D torus. In Table \ref{table3}.Bottom we show the results of our method applied to all the bars of the persistence barcode without separating by dimensions. We can see in this case that we obtain as topological features the  length of the bar representing the 
connected component, the ones representing the two tunnels and the one representing the void.
In Table \ref{table3}.Top  (resp. Table \ref{table3}.Middle),
we show the results of our method applied to the bars of the $0$-dimensional (resp. $1$-dimensional) persistence barcode.
Observe that the results are consistent with the ones obtained in Table \ref{table3}.Bottom.

\begin{table}
 \centering
\begin{tabular}{c}
\begin{tabular}{|c|c|c|c|c|}
\hline
$\text{Iteration}$& $n'$& $Q$&   $E(L')/E(M')$ &$\alpha$ \\
 \hline
 1 & 400& 47 & 0.993675 & 0.0521053\\
 2 & 286&34  & 0.995459 & 0.0521053\\
 3 & 154& 19 & 0.993659 & 0.0521053\\
 4& 40& 5 & 0.977194 & 0.0521053\\
 \hline\hline
 $\ell_i$  &$\ell'_i$ & C& $E(L'_i)/E(M')$ & \text{Feature?} \\
 \hline
 1.9 & 0.0270463 & 0.997193 & 0.977234 & \text{yes} \\
 0.396 & 0.0270402 & 0.996476 & 0.977294 & \text{no} \\
 0.387 & 0.0270339 & 0.996296 &0.977357 & \text{no} \\
  \dots &   \dots  &  \dots &  \dots &  \dots\\
  \hline
\end{tabular}
\\
\\
\begin{tabular}{|c|c|c|c|c|}
\hline
 $ \text{Iteration}$& $n'$& $Q$& $E(L')/E(M')$ &$\alpha$ \\
 \hline
 1 & 177&5  &0.893511& 0.00587851\\
 2 & 50& 2 &0.904135& 0.00587851\\
 3 & 6& 1 &0.796461& 0.00587851\\
   \hline\hline
 $\ell_i$  &$\ell'_i$ & C& $E(L'_i)/E(M')$ & \text{Feature?} \\
  \hline
 1.531 & 0.264348 & 1.22428 & 0.822957 & \text{yes} \\
 1.531 & 0.242872 & 0.791795 & 0.875366 & \text{yes} \\
 0.27 & 0.221058& 0.751341 &0.933577 & \text{no} \\
 0.261 & 0.198549 & 0.703848 &1. & \text{no} \\
 \dots &   \dots  &  \dots &  \dots &  \dots\\
  \hline
\end{tabular}
\\
\\
\begin{tabular}{|c|c|c|c|c|}
 \hline
$\text{Iteration}$& $n'$& $Q$&   $E(L')/E(M')$ &$\alpha$ \\
 \hline
  1 & 578&47 &0.969019& 0.00473684\\
 2 & 429& 9 &0.988556& 0.00473684\\
 3 &306 & 7 &0.989757& 0.00473684\\
 4 &179 &4  &0.985682& 0.00473684\\
 5 & 43& 1 &0.95329& 0.00473684\\
 \hline\hline
   $\ell_i$  &$\ell'_i$ & C& $E(L'_i)/E(M')$ & \text{Feature?} \\
 \hline
 1.9 & 0.0268932 & 1.05166 & 0.961183 & \text{yes} \\
 1.531 & 0.0259429 & 1.05987 & 0.970746 & \text{yes} \\
 1.531 & 0.0253107 & 1.04887 & 0.977304 & \text{yes} \\
 1.234 & 0.0253068 & 0.996996 & 0.977345 & \text{yes} \\
 0.396 & 0.025301 & 0.99627 & 0.977407 & \text{no} \\
 0.387 & 0.0252948 & 0.996079 & 0.977471 & \text{no} \\
  \dots &   \dots  &  \dots &  \dots &  \dots\\
  \hline
\end{tabular}
\end{tabular}
\caption{Results of our method applied to the persistence barcodes  of the  Vietoris-Rips filtration obtained from $400$ points sampled from  a 3D torus. In the table on the top, only bars in the $0$-dimensional pesistence barcode are taken into account. In the table on the middle, only 
bars in the $1$-dimensional pesistence barcode are considered. In the  table on the bottom,   the bars of persistence barcodes of dimension $0$, $1$ and $2$ are considered altogether. }
\label{table3}
\end{table}

\begin{figure}[t!]
	\centering
		\includegraphics[width=12cm]{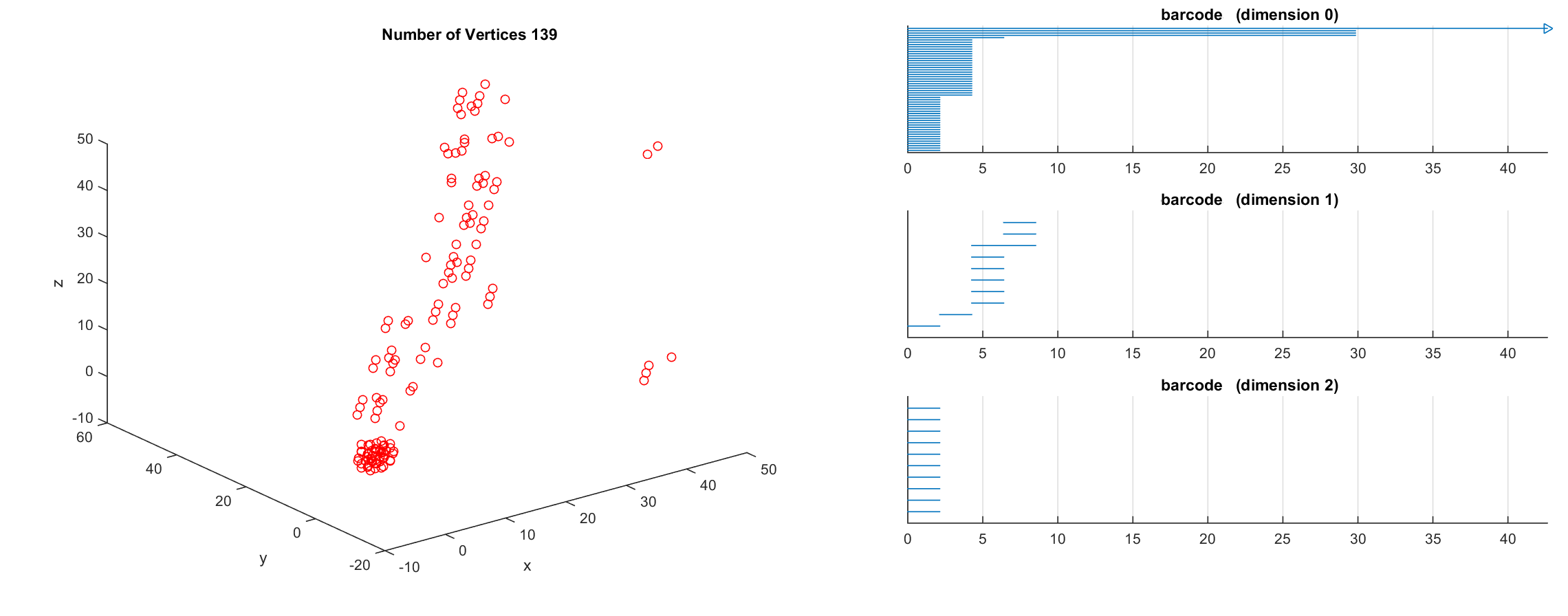}
	\caption{Left: A PCD $V$ in $\mathbb{R}^3$ composed by $139$ points. Right, from top to bottom: $0$-, $1$- and $2$-dimensional persistence barcode of $VR_V$. 
    }
    \label{fig:ejemplo3}
\end{figure} 

\begin{table}
 \centering
 \begin{tabular}{c}
\begin{tabular}{|c|c|c|c|c|}
\hline
 $ \text{Iteration}$& $n'$& $Q$& $E(L')/E(M')$ &$\alpha$ \\
 \hline
 1 & 76& 9 &0.857307& 0.05\\
  \hline\hline
 $\ell_i$  &$\ell'_i$ & C& $E(L'_i)/E(M')$ & \text{Feature?} \\
  \hline
 42.4484 & 0.023405 & 1.06815 & 0.869108 & \text{yes} \\
 31.8363 & 0.0219565 & 1.08028 & 0.883896& \text{yes} \\
 31.8363 & 0.0202248 & 1.09583 & 0.902912 & \text{yes} \\
 31.8363 & 0.0181347 & 1.11616 & 0.928160 & \text{yes} \\
 31.8363 & 0.0181326 & 0.997812 & 0.928187 & \text{yes} \\
 4.24484 & 0.0181304 & 0.997733 & 0.928215 & \text{no} \\
 4.24484 & 0.0181281 & 0.99765 & 0.928244 & \text{no} \\
  \dots &   \dots  &  \dots &  \dots &  \dots\\
  \hline
\end{tabular}
\\
\\
\begin{tabular}{|c|c|c|c|c|}
\hline
 $ \text{Iteration}$& $n'$& $Q$& $E(L')/E(M')$ &$\alpha$ \\
 \hline
1 & 98& 12 &0.870766& 0.05\\
2 & 30&4  &0.852821& 0.05\\
  \hline\hline
 $\ell_i$  &$\ell'_i$ & C& $E(L'_i)/E(M')$ & \text{Feature?} \\
  \hline
 42.4484 & 0.023405 & 1.06815 & 0.861170 & \text{yes} \\
 31.8363 & 0.0219565 & 1.08028 & 0.8725486& \text{yes} \\
 31.8363 & 0.0202248 & 1.09583 & 0.888921 & \text{yes} \\
 31.8363 & 0.0181347 & 1.11616 & 0.914344 & \text{yes} \\
 31.8363 & 0.0181326 & 0.997812 & 0.916294 & \text{yes} \\
 4.24484 & 0.0181304 & 0.997733 & 0.918325 & \text{no} \\
 4.24484 & 0.0181281 & 0.99765 & 0.920440 & \text{no} \\
  \dots &   \dots  &  \dots &  \dots &  \dots\\
  \hline
\end{tabular}
\end{tabular}
\caption{Results of our method applied to the PCD showed in Fig. \ref{fig:ejemplo3}.}
\label{table6}
\end{table}

Finally, consider a PCD $V$ of $76$ points in $\mathbb{R}^3$.
The barcodes (separated by dimension) computed from the Vietoris-Rips filtration associated to $V$ are showed in Fig. \ref{fig:ejemplo3}.Right. In Table \ref{table6}.Top, we have applied our method to the $0$-barcode  and in
 Table \ref{table6}.Bottom, the method is applied to all the bars of the barcode without differentiating by dimension.  
We can see that in both cases we obtain $5$ topological features. The idea is that, compared to the longest bars in the $0$-barcode, the bars in the $1$- and $2$-barcodes are considered noise:  The two tunnels and the one representing the void.
This example shows that, since bars in higher dimensions than $0$ are noise,   the results obtained with our method are independent on applying it on the whole set of bars of the persistence barcode  or separating the bars by dimension.

\section{Conclusions and future work}\label{conclusion}
 Vietoris-Rips complexes are a fundamental tool in topological data analysis, they allow to build a topological space from  higher dimensional data-set embedded in a metric space \cite{nanda2014simplicial}. The resulting complex is then studied by persistent homology. In order to provide a summary of the information provided by persistent homology new statistics have been defined. Among the statistics, we put our focus on a Shannon-like entropy that is known as persistent entropy. Persistent entropy records how much is ordered the construction of a topological space. In this paper, we discuss several properties of the persistent entropy when it is computed on  the persistence barcode of a given Vietoris-Rips filtration. The first property demonstrates the relations between  persistent entropy and the bottleneck distance, that is a well known measures for comparing persistence barcodes. This is a preliminary results for assuring that persistent entropy is a stable measure for dealing with Vietoris-Rips complexes. Moreover, the computation of persistent entropy is less computational expensive with respect to the bottleneck distance. Because the construction of Vietoris-Rips depends on the choice of the upper bound of a parameter, we identify a new quantity that can be used for and we hope this can be a signpost by the reader when he/she starts to investigate the construction of the Vietoris-Rips complexes. By introducing this quantity we are able to define a new methodology based on persistent entropy for identifying which are true topological features and which must be considered noisy topological features. We apply the methodology on a couple of examples. Briefly, the method is an iterative algorithm that at the $i$-th step  replaces the first $i$ bars by the same number of bars but with the length that maximizes the entropy. This way we ``neutralize" the effect of such $i$ bars and we can deduce if the bar at position $i$ is a topological feature or not.
 \\
As future works we are planning to extend the properties to the Witness complexes, that are roughly speaking a way for computing the Vietoris-Rips complexes from very large data-set. This will allow us to use our method for studying biological data  as well as RNA data for differentiating healthy cells from unhealthy cells \cite{davie2015discovery,nasuti2016metal}. We argue the method will let to highlight the topological features that are formed by the most relevant genes associated to pathologies.

%
%

\bibliographystyle{elsarticle-num}

\end{document}